\documentclass[12pt,final,leqno]{amsart}

\usepackage{t1enc,amsfonts,latexsym,amsmath,verbatim,amssymb,epsfig,graphics,psfrag}

\def\B{{\mathcal B}}

\def\H{{\mathcal H}}

\def\N{{\mathbb N}}

\def\C{{\mathbb C}}

\def\R{{\mathbb R}}
\def\Z{{\mathbb Z}}

\def\M{{\mathcal M}}
\def\L{{\mathcal L}}

\def\f{{\mathcal F}}
\def\z{{\mathcal Z}}
\def\s{{\mathcal S}}

\def\p{{\mathcal P}}

\newcommand{\ii}{\mathsf{i} }
\newcommand{\ee}{\mathsf{e} }

\newcommand{\cl}{\mathsf{cl} }
\newcommand{\supp}{\mathsf{supp}~ }

\newcommand{\Span}{\mathsf{Span}~ }

\numberwithin{equation}{section}

\theoremstyle{plain}
\newtheorem{theorem}{Theorem}[section]
\newtheorem{proposition}[theorem]{Proposition}
\newtheorem{lemma}[theorem]{Lemma}
\newtheorem{corollary}[theorem]{Corollary}
\newtheorem{definition}  [theorem] {Definition}

\theoremstyle{definition}
\newtheorem{ex}[theorem] {Example}
\newtheorem{example}[theorem] {Example}

\title{On linear differential equations with infinitely many derivatives.}

\author[M. Carlsson, H. Prado, E.G. Reyes]{Marcus Carlsson, Humberto Prado 
and Enrique G. Reyes}


\begin{document}

\maketitle

\begin{abstract}
Differential equations with infinitely many derivatives, sometimes also referred to as ``nonlocal'' differential equations, appear frequently in branches of modern
physics such as string theory, gravitation and cosmology. The goal of this paper is to show how to properly interpret and solve such equations, with a special focus 
on a solution method based on the Borel transform. This method is a far-reaching generalization of previous approaches (N. Barnaby and N. Kamran, Dynamics with infinitely 
many derivatives: the initial value problem. {\em J. High Energy Physics} 2008 no. 02, Paper 008, 40 pp.; P. G\'orka, H. Prado and E.G. Reyes, Functional calculus via 
Laplace transform and equations with infinitely many derivatives. {\em Journal of Mathematical Physics} 51 (2010), 103512;  P. G\'orka, H. Prado and E.G. Reyes, The 
initial value problem for ordinary equations with infinitely many derivatives. {\em Classical and Quantum Gravity} 29 (2012), 065017).
In particular we reconsider generalized initial value problems and disprove various conjectures found in the modern 
literature. We illustrate various phenomena that can occur with concrete examples, and we also treat efficient implementations of the theory.
\end{abstract}

\section{Introduction}

Differential equations with infinitely many derivatives appear frequently in modern
physics, and special cases of such ``nonlocal" equations have been extensively
studied in string theory, quantum gravitation theory and cosmology, see for instance
\cite{AV,B,BBC,BBOR,BK1,BK2,CMN1,CMN,D,EW,Moeller,V,VV} and also \cite{Ra,Ta,W}. 
The mathematical study of such equations began over a century ago (see e.g. \cite{ritt}), and there are several methods for finding solutions, most of these based 
on the Fourier and Laplace transforms (see for instance \cite{BK1,BK2,GPR_JMP,GPR_CQG}). Our main contribution in this paper is to generalize these classical approaches 
and to present a method relying on the Borel transform, see Definition 2.2 below. 
Interestingly, this transform was already used in this context in the 1930's by R.D. Carmichael \cite{Car1,Car} and it has reappeared occasionally, see e.g. \cite{Gruman}, 
but we believe it has not received the attention it deserves. We will show that the Borel transform method has several benefits due to its generality and simplicity. 
In addition to solving nonlocal equations in a framework which is more general than the one provided by Fourier and Laplace transform methods, we will show that this method 
combined with the theory of entire functions can be used to give counterexamples of some conjectures found in the modern physics literature, especially 
concerning generalized initial value problems \cite{B,GPR_CQG}. We also include a comparison of our method with other proposals which have appeared in the literature 
\cite{dybinskii, H,H1,Tran}, as well as many examples, providing insight into pros and cons with the various techniques and demonstrating what phenomena 
to expect.

We consider the equation

\begin{equation}\label{a1w}
\phi\left(\frac{d}{dx}\right)f(x)=g(x)
\end{equation}
where $f$ is an unknown function on $\R$, $g$ is a given function
on $\R$, and the symbol $\phi(z)=\sum_{k=0}^\infty a_k z^k$ is an
analytic function. We stress the fact that such equations can be solved using a number of different approaches, in addition to the one mentioned above based 
on the Borel transform. For example, given that $\phi$ is ``nice'', we could consider (\ref{a1w}) as a particular case of a pseudo-differential equation and use 
the distribution theory approach developed by L. Schwartz, L. H\"{o}rmander and others \cite{H}. Alternatively, there exists a theory developed by 
Yu. A. Dubinskii and his collaborators \cite{dub1,dybinskii,Tran}, which allows more general symbols $\phi$'s, but which also relies on the Fourier 
transform. For equations on the half-line there are also Laplace transform methods, see e.g. \cite{BK1,GPR_JMP,GPR_AHP}. However, the suitability of these methods
for a particular application ---including the question of how to interpret the left hand side of (\ref{a1w})--- depends on $\phi$ or on conditions 
on $f$ and $g$ imposed by the physical problem at hand (see for instance \cite{B,BK1,Moeller}). By comparison, we have found that the Borel transform approach
turns out to be quite flexible and therefore of very general applicability.

In this paper we will consider (\ref{a1w}) for $f,g$ in the class $Exp$, that is, restrictions to $\R$ of entire functions of exponential type. As we 
shall see, there are several benefits of working with the class $Exp$. First of all $Exp$ is a rather large space, containing e.g. all band limited functions in 
$L^2(\R)$ as well as all exponential functions. We note that when $\phi$ is a polynomial, i.e. when (\ref{a1w}) reduces to an ordinary differential 
equation, then the span of exponential functions of the form $x^m\ee^{\lambda x}$, (where $\phi(\lambda)=0$ and $m$ is less than the multiplicity of 
$\lambda$), equals the set of homogenous solutions to (\ref{a1w}). A pleasant fact is that these functions are in $Exp$, which is not the case for 
$L^2(\R)$ or more general classes of Sobolev type spaces. Thus, the theory for existence and uniqueness of solutions becomes particularly simple and 
natural in $Exp$. Relying on the Borel transform, we give a simple explicit formula for the solution of (\ref{a1w}), given $g\in Exp$. We will also 
show that the homogenous solutions are obtained by a direct generalization of the situation for ODE's. Summing up, if we know the Borel transform 
of $g$ and the zeroes of $\phi$, all solutions to (\ref{a1w}) in $Exp$ are explicitly given. This is done in Sections \ref{secExp}, \ref{secty}, 
\ref{secPart}, \ref{secHom}. The issue of actually computing the Borel transform in a concrete application is addressed in Section \ref{secImp}. 
Yet another good feature of working with $Exp$ is that there can be no ambiguity as to how $\phi\left(\frac{d}{dx}\right)f(x)$ should be interpreted, 
since we shall see that 

\begin{equation}\label{a2w}
\phi\left(\frac{d}{dx}\right)f=\lim_{K\rightarrow\infty}\sum_{k=0}^K a_k \frac{d^kf}{dx^k},\quad f\in Exp
\end{equation}
with respect to uniform convergence on compact sets. This is in sharp contrast with e.g. the space $C^\infty_{c}$ of infinitely differentiable functions 
with compact support, for we shall prove that the right hand side diverges in $L^1_{loc}$ for all $f\in C^\infty_{c}$ and $\phi(z)=\ee^{z^2}$ 
(Example \ref{e1}). This rules out convergence in any weighted $L^p$-space, and with the same method one can also show that it does not converge in any 
classical Sobolev space. This is remarkable since $\ee^{(\frac{d}{dx})^2}$ does have a physically relevant interpretation in e.g. $L^2(\R)$ (convolution 
with the heat kernel) which arises when considering the heat equation on the line. Since a subset of $Exp$ is dense in $L^2(\R)$ (we can take, for 
instance, the set of Hermite functions) this physically relevant interpretation of $\ee^{(\frac{d}{dx})^2}$ on $L^2(\R)$ is easily derived from the 
general theory of this paper. We refer to Example \ref{b1} for more details. Now, the function $\phi(z)=\ee^{z^2}$ is special because it is a bounded 
function on $\ii\R$, which is an underlying fact for the phenomena discussed above. We shall see in Section \ref{secty} and \ref{sectd} that in order 
to have a general theory for solving (\ref{a1w}) for arbitrary entire functions $\phi$, we are forced to work with functions $f,g$ that are restrictions 
to $\R$ of entire functions, as $Exp$.

After the general theory has been developed in Sections \ref{secExp}, \ref{secty}, \ref{secPart}, \ref{secHom}, we consider particular cases in Sections 
\ref{sectd} and \ref{secCE}. In particular, we show that equations of convolution type can be reformulated as (\ref{a1w}), and hence solved via the 
Borel transform. Section \ref{secREW} is devoted to showing the power of the present approach in comparison with the other methods mentioned initially.  
As we shall see, the various interpretations of the operator $\phi(\frac{d}{dx})$ coincide on common domains of definition. Thus, the question of which method 
to chose will depend on the symbol $\phi$ and on the particular function space one wishes to work in, as well as on the physical nature from which the equation 
has been deduced. 

Section \ref{secini} sheds further light on initial value problems, see \cite{GPR_JMP, GPR_CQG}. A recurring issue in the physics literature is how to 
appropriately impose initial value conditions that uniquely specify the solution, in analogy with the theory of ODE's (\cite{B,BK1,Moeller}). Here, we
use essentially the fact that we are working with real analytic functions, for it allows us to use deep results from the theory of entire functions 
of finite order, to (dis)prove some natural conjectures \cite{B,Moeller}. In particular, we shall show that if the exponent of convergence of the zeroes 
of $\phi$ is $>1$, then any function in $L^2(I)$, in which $I$ is a finite interval, is arbitrarily close in $L^2(I)$ to a solution of the equation 
$\phi\left(\frac{d}{dx}\right)f(x)=0$. 
See Proposition \ref{tgb} for definitions and a precise statement. The paper ends with Section \ref{secImp} treating the issue of how to efficiently 
implement stable solvers based on the methods presented in this paper.

\section{Functions of exponential type}\label{secExp}

We provide a review of key properties of functions of exponential type. This is an extensively studied subject, and many more details can be found e.g. in 
\cite{Boas,Levin}.

\begin{definition} \label{d1}
An entire function $F$ is said to be of exponential type if there exists a $C,\tau>0$ such that $$|F(z)| \leq C \ee^{\tau|z|}.$$
The number $\tau$ is called an exponential bound of $F$, and the infimum of all possible $\tau$'s is called the exponential type of $F$. 
The set of all entire functions of exponential type will be denoted by $Exp(\C)$.
\end{definition}

Examples of functions of exponential type include $\sin(z),~\cos(z)$ as well as all polynomials, but not e.g. $\ee^{z^2}$.
We will denote by $Exp(\R)$ the space of all functions $f$ such that there exists an
$F\in Exp(\C)$ with $f=F|_{\R}$ ($F$ restricted to $\R$). Note that $Exp(\R)$ contains e.g. all band-limited functions in $L^2(\R)$, i.e. all 
functions whose Fourier transform has compact support. This is not hard to deduce by applying the Fourier inversion formula, but it also follows 
from Proposition \ref{p1} and formula (\ref{fourier}) below. When there is no risk of confusion, we will make no distinction between $f$ on $\R$ 
and its extension to $\C$, and simply write $f\in Exp$.

The class $Exp$ will be the domain of the nonlocal differential operators we consider in this paper,
while we will pose no restriction on the class of symbols $\phi$ appearing in our equations, except that they should be entire functions. 
We can give a concrete integral representation of the functions in $Exp$ via the Borel transform which we define as follows:

\begin{definition}\label{d5}
Suppose that $f(z)=\sum_{k=0}^\infty b_k z^k$ belongs to $Exp(\C)$. The
Borel transform of $f$ is defined as
\begin{equation}\label{a10}
\B(f)(z)=\sum_{k=0}^\infty \frac{k!\,b_k}{z^{k+1}}\; ,
\end{equation}
\end{definition}

If $\tau$ is the exponential type of $f$, it is a classical fact (see e.g. Theorem 5.3.1, \cite{Boas}) that the series in (\ref{a10}) converges for all 
$|z|>\tau$, whereas this is false for all smaller discs. Even more, if we let $S$ be the \textit{conjugate diagram of $\B(f)$}, i.e. the smallest convex
set contained in $\{z : |z|\leq \tau\}$ such that $\B(f)$ extends by analyticity outside $S$, then a classical theorem due to Polya tells us that $S$ can 
be characterized in terms of the growth of $f$ along rays emanating from $z=0$. We refer to \cite[Chapter 5]{Boas} for the details, but we remark that this 
is closely connected with the following alternative definition of $\B(f)$ via the Laplace transform $\L$; given $z$ with $|z|>\tau$ we have
\begin{equation}\label{a11}
\B(f)(r\ee^{i\theta})=\ee^{-i\theta}\int_{0}^\infty f(t\ee^{i\theta})\ee^{- rt  }~dt=\ee^{-i\theta}\L\left(f(\cdot ~\ee^{i\theta})\right).
\end{equation}
We note that both (\ref{a10}) and (\ref{a11}) give concrete formulas for calculating the Borel transform
of an explicitly given function $f\in Exp(\C)$. If one only knows $f$ on $\R$, neither of the above formulas can be evaluated, but we will see in 
Section \ref{secImp} how this can be circumvented. The following formula, which says how to
recover $f$ from $\B(f)$, is crucial for our theory. Given any $R>\tau$ we have
\begin{equation}\label{a12}
f(z)=\int_{|\zeta|=R} \ee^{z\zeta}\B(f)(\zeta)~\frac{d\zeta}{2\pi \ii},
\end{equation}
(Theorem 5.3.5, \cite{Boas}).

\begin{ex}\label{b2}
Let $f(x)=\ee^x$. Then $f(z)=\sum_{k=0}^\infty \frac{z^k}{k!}$ and $f(z)$ is obviously of exponential type $\tau =1$. We have, 
$$
\B(\ee^{x})=\sum_{k=0}^\infty \frac{1}{\zeta^{k+1}}=\frac{1}{\zeta}\frac{1}{1-1/\zeta}=\frac{1}{\zeta-1} \; , 
$$
which converges for $|\zeta|>1$, as predicted by the theory. In this case, we note that the convex set
$S$ given by Polya's theorem mentioned above is significantly smaller than $\{\zeta:~|\zeta|\leq 1\}$; obviously
$S=\{1\}$. Now, we can invert $\B\,$: for $R>1$, (\ref{a12}) implies that
\begin{equation}\label{gt}
\ee^x =
\int_{0}^{2\pi} \ee^{xR\ee^{\ii\theta}}\frac{R\ee^{\ii\theta}}{R\ee^{\ii\theta}-1}~\frac{d\theta}{2\pi}
\end{equation}
which can be verified directly by expanding $\ee^{xR\ee^{\ii\theta}}$ and $(1-\ee^{-\ii\theta}/R)^{-1}$ in power series, integrating
and collecting non-zero terms.
\end{ex}

$\Box$

To work with expressions of the type (\ref{a12}), we now introduce an extension of the Fourier-Laplace transform \cite{H,H1}. It will turn out 
convenient (but it is not strictly necessary) to also work with distributions. To define distributions on $\C$ we identify it with $\R^2$ via 
$\zeta=\xi+\ii \eta \leftrightarrow (\xi,\eta)$. We refer to \cite{H}, Sec 2.3 for basic results on distributions of compact support on $\R^n$. 
Recall in particular that any distribution $d$ of compact support automatically has finite order $N$, and that one then can find $C>0$ such that 
\begin{equation}\label{et6}
|\langle d,\tau \rangle|\leq C\sum_{n_1+n_2\leq N} \|\partial_\xi^{n_1}\partial_\eta^{n_2} \tau(\zeta)\|_{L^\infty}
\end{equation}
for all test functions $\tau$.

\begin{definition}
If $d$ is a distribution with compact support in $\C$, we set
\begin{equation}\label{a13}
\p(d)(z) = \langle \ee^{z\zeta} , d \rangle \; ,\quad z\in\C
\end{equation}
where $\zeta$ represents the independent variable on which $d$ acts. If $d=\mu$ is a measure, Equation $(\ref{a13})$ clearly reduces to 
$$\p(\mu)(z)=\int_{\C} \ee^{z\zeta}~d\mu(\zeta).$$
\end{definition}

The following proposition shows that the image of $\p$ equals $Exp$.

\begin{proposition}\label{p1}
If $d$ is a distribution with support in $\{\zeta:~|\zeta|< R\}$, then $\p(d)\in Exp$ and 
\begin{equation}\label{eq4}
|\p(d)(z)|\leq C\ee^{R|z|}
\end{equation} for some $C>0$. Conversely, given any $f\in Exp$ we can pick a measure $\mu_f$ such that $f=\p(\mu_f)$. The measure $\mu_f$ is not unique, 
on the contrary, if $f$ has exponential type $\tau$ and $R>\tau$ is arbitrary, $\mu_f$ can be chosen to have support on $\{\zeta:~|\zeta|=R\}$.
\end{proposition}
\begin{proof}
If $d=\mu$ is a complex measure (of finite variation) the first part follows as
$$|\p(\mu)(z)| \leq \int_{|\zeta|\leq R} \ee^{|z|R}~d|\mu|(\zeta)\leq \ee^{|z|R}\|\mu\|.$$ The corresponding result for distributions is a bit more 
complicated. Let $\chi\in C^\infty(\C)$ be identically $1$ on the support of $d$ and identically $0$ outside $\{\zeta:~|\zeta|< R_1\}$, where $R_1<R$. By 
(\ref{eggg}) it then follows that 
$$|\p(d)(z)|=\left|\left\langle d(\zeta),\ee^{z\zeta}\right\rangle\right|=\left|\left\langle d(\zeta),\chi(\zeta)\ee^{z\zeta}\right\rangle\right|\leq C\sum_{n_1+n_2\leq N} \|\partial_\xi^{n_1}\partial_\eta^{n_2}\chi(\zeta)\ee^{z\zeta}\|_{L^\infty},$$ 
where $C$ is a constant, $N$ is the order of $d$ and the inequality follows by the definition of distributions, see (2.3.1) in \cite{H}. From this it 
easily follows that $$|\p(d)(z)|\leq C_1 (1+|z|^N)\ee^{R_1|z|}$$ for some $C_1>0$. Since $R_1<R$, the inequality (\ref{eq4}) follows. This shows that 
$Im \p\subset Exp$. The reverse inclusion follows by realizing that the integral representation (\ref{a12}) can be rewritten as $f=\p(\mu_{f,R})$ with
the measure $\mu_{f,R}$ defined as 
\begin{equation} \label{mug}
\mu_{f,R}(E)=\int_{\{\theta:~R\ee^{\ii\theta}\in E\}} \B(f)(R\ee^{\ii\theta})~R\ee^{\ii\theta}~\frac{d\theta}{2\pi},\quad E \subset \C
\end{equation}
\end{proof}

Note that for any distribution $d$ with compact support, Equation (\ref{a12}) implies that the function $\p(d)$ can be equivalently
represented by a measure supported on a circle. Despite this, we will see in Section \ref{sectd} that it can be much easier to work 
directly with distributions.

\begin{example}\label{b3}
Continuing Example \ref{b2}, we have $$\ee^x=\p(\delta_1)(x)$$
where $\delta_1$ is the Dirac measure with support at $\{1\}$. On the other hand by (\ref{gt}) and (\ref{mug}) we have that
$\ee^x=\p(\mu_{\ee^x})$ with \begin{equation*}
\mu_{\ee^x}(E)=\int_{\{\theta:~R\ee^{\ii\theta}\in E\}} \frac{R\ee^{\ii\theta}}{R\ee^{\ii\theta}-1}~\frac{d\theta}{2\pi},\quad E \subset \C
\end{equation*} and $R>1$.
\end{example}

$\Box$

\begin{example}\label{b4}
the function 
$f(x)=\cos x$ is a function of exponential type 1 since
$$
\cos x=\frac{1}{2}\left(\ee^{\ii x}+\ee^{-\ii x}\right)=
\p\left(\frac{1}{2}(\delta_{\ii}+\delta_{-\ii})\right) \; ,
$$
in which $\delta_{\pm \ii}$ are Dirac measures with support at $\pm \ii$. 
The Borel transform is easily calculated to
$\B(\cos x)=\frac{1}{2}\left(\frac{1}{\zeta-\ii}+\frac{1}{\zeta+\ii}\right)=\frac{\zeta}{\zeta^2+1}$, and
hence for $R>1$ we have
\begin{equation*}
\cos x=\int_{0}^{2\pi}\ee^{xR\ee^{\ii\theta}}\frac{R\ee^{\ii\theta}}{R\ee^{2\ii\theta}+1}~\frac{d\theta}{2\pi}
\end{equation*}
\end{example}

$\Box$

More generally, letting $Pol$ denote the set of all polynomials, we have that any function of the form
$$
f(x)=\sum_{finite} p_k(x)\ee^{\zeta_k x},\quad p_k\in Pol,~\zeta_k\in\C
$$
is in $Exp$ and can therefore be expressed as (\ref{a12}). However, we remark that a more direct representation is given
by letting $d$ be the distribution
$\sum p_k(-\partial_{\xi})\delta_{\zeta_k}$, (where the independent variable is $\zeta=\xi+\ii\eta$), and using (\ref{a13}). Indeed, 
\begin{align*}
&\p(d)(x)=\left\langle \ee^{x(\xi+\ii\eta)}, \sum p_k(-\partial_{\xi})\delta_{\zeta_k}\right\rangle=\\
&=\sum\left\langle p_{k}(\partial_{\xi})\ee^{x(\xi+\ii\eta)}, \delta_{\zeta_k}\right\rangle=\sum\left\langle p_{k}(x)\ee^{x(\xi+\ii\eta)}, \delta_{\zeta_k}\right\rangle=f(x)
\end{align*}
by standard calculation rules for distributions.

The transform $\p$ is an extension of both the Fourier and Laplace transforms, since the former appears if
we restrict $\p$ to measures $\mu$ supported on $\ii\R$ and the latter if we restrict $\p$ to measures
supported on $\R^-$. More concretely, if $u:\R^+\rightarrow\C$ has compact support and we define the distribution 
$u(-\xi)\delta_0(\eta)$ via $$\left\langle u(-\xi)\delta_0(\eta),\tau(\xi,\eta)\right\rangle=\int u(-\xi)\tau(\xi,0)~d\xi,$$ we get
\begin{equation}\label{laplace}\begin{aligned}
&\p\big(u(-\xi)\delta_0(\eta)\big)(z)=\left\langle u(-\xi)\delta_0(\eta),\ee^{(\xi+\ii\eta) z}\right\rangle=\\&=\int \ee^{\xi z} u(-\xi)~d\xi=\int \ee^{-\xi z} u(\xi)~d\xi=\L(u)(z),
\end{aligned}\end{equation}
where $\L$ denotes the Laplace transform. Clearly $u(-\xi)\delta_0(\eta)$ can be identified with the measure $\mu$ given by
$$\mu(E)=\int_{E\,\cap\, \R} u(-\xi)~d\xi,\quad E\subset\C.$$
Similarly, let $\f$ be the (unitary) Fourier transform. Given $u:\R\rightarrow\C$ with compact support we have 
\begin{equation}\label{fourier}\p\big(\delta_0(\xi)u(-\eta)\big)(z)=\int \ee^{\ii \eta z} u(-\eta) ~d\eta=\sqrt{2\pi} \f(u)(z)
\end{equation} 
and a short argument shows that $\delta_0(\xi)u(-\eta)$ can be identified with the measure
$$\mu(E)=\int_{\ii E\,\cap\, \R} u(\eta)~d\eta,\quad E\subset\C.$$

One could thus consider $\p$ as a mere extension of the Fourier transform and call it something like
"the extended Fourier-Laplace transform" (see \cite{H,H1}). However, we avoid this since both the Fourier and Laplace
transforms are intimately connected with their respective domains of definition, with associated inverses, 
etc. The transform $\p$ is different from the Fourier-Laplace transform in that its argument $d$ is allowed to have support on $\C$ 
as opposed to $\R$, which significantly alters most properties.

\section{Defining the functional calculus of $\frac{d}{dx}$}\label{secty}

In this section we return to the differential equation (\ref{a1w}), and show that the left hand side is well
defined for any $f\in Exp$ in accordance with the limit (\ref{a2w}) (with respect to the topology of uniform convergence on compacts). 
We recall that if $d$ is a distribution and $\phi\in C^\infty(\C)$, we can define a new distribution $\phi d$ via 
$\langle \phi d,\tau\rangle=\langle  d, \phi\tau\rangle$, where $\tau$ represents a test function. If $d=\mu$ is a measure, $\phi\mu$ 
is also a measure and $\int_{\C}\tau~d(\phi\mu)=\int_{\C} \tau \phi~d\mu.$
These facts are standard, see e.g. \cite{cohn}. The main theorem of this section is as follows.

\begin{theorem}\label{c1}
Given an entire function $\phi(z)=\sum_{k=0}^\infty a_k z^k$, $f\in Exp$ and a distribution $d_f$ such that $f=\p(d_f)$,
we have
\begin{equation}\label{a2'}
\lim_{K\rightarrow\infty}\sum_{k=0}^K a_k
\frac{d^kf}{dx^k}=\p(\phi d_f)
\end{equation}
uniformly on compacts. In particular the limit $(\ref{a2'})$ exists, and we will henceforth denote it by
$\phi\left(\frac{d}{dx}\right)f$.
\end{theorem}

\begin{proof}

If $d_f=\mu_f$ is a measure with
compact support, then the identity
\begin{equation}\label{eggg}\frac{d}{dz}f=\frac{d}{dz}\p(d_f)=\p(\zeta d_f)
\end{equation} 
follows by the Dominated Convergence Theorem. The formula is true also in the general case, and the proof builds on the same estimate we used 
in Proposition \ref{p1}. We omit the details. By repeated use of (\ref{eggg}) we get
\begin{equation*}
\sum_{k=0}^K a_k \frac{d^k\,f}{dz^k}-\p
(\phi\mu_f)  =\p\left(\left(\sum_{k=0}^K a_k \zeta^k-\phi(\zeta)\right)d_f\right)=\left\langle d_f,\left(\sum_{k=0}^K a_k \zeta^k-\phi(\zeta)\right)\ee^{z\zeta}\right\rangle
\end{equation*}

Since $\sum_{k=0}^K a_k \zeta^k$ converges uniformly on compacts to $\phi(\zeta)$, we have that
\begin{align*}
&\lim_{K\rightarrow\infty}\sum_{k=0}^K a_k \frac{d^k\,f}{dz^k}(z)
=  \lim_{K\rightarrow\infty}\sum_{k=0}^K a_k
\frac{d^k\,}{dz^k}\p(\mu_f)(z) = \lim_{K\rightarrow \infty}
\sum_{k=0}^K a_k \frac{d^k}{dz^k} \int_\C  \ee^{z\zeta} d\mu_f(\zeta)= \\
& =\lim_{K\rightarrow \infty} \int_\C  \sum_{k=0}^K a_k \zeta^k
\ee^{z \zeta} d\mu_f(\zeta) = \int_\C  \phi(\zeta) \ee^{z \zeta} d\mu_f(\zeta) = \p
(\phi\mu_f)(z).
\end{align*}
It is also easy to see that the limit is uniform on compact sets. If $d_f$ is a distribution, the result follows by a similar modification as in 
Proposition \ref{p1}.
\end{proof}

For instance,
\begin{equation}\label{e4}
\ee^{y\partial_x} u(x) = \int \ee^{x\zeta} \ee^{y \zeta } d \mu_u (\zeta) = \int \ee^{(x+y)\zeta} d \mu_u (\zeta) = u(x+y)\; .
\end{equation}

If we allow complex $y'$s, (\ref{e4}) indicates that if one wishes to develop a functional calculus for $\frac{d}{dx}$ which includes all entire 
functions, or at least all $\phi$'s of the form (\ref{e4}), then we are forced to operate on a space of functions that are restrictions of entire 
functions. Hence, most classical function spaces like $C(\R)$ or $L^2(\R)$ are not suitable. This statement is further supported by the next example, 
which gives a particular $\phi$ such that the convergence of (\ref{a2w}) is not compatible with functions of compact support. We denote by 
$L^1_{loc}(\R)$ the set of all functions whose restriction to any compact interval $I$ is integrable.

\begin{ex}\label{e1}
\textit{Claim: If $f\in C^\infty(\R)$ has compact support, the sequence $f_K=\sum_{k=0}^K
\frac{\partial_x^{2k} f}{k!}$ does not converge in $L^1_{loc}(\R)$ (hence neither does it converge in any weighted
$L^p(\R)$-space, by H\"{o}lder's inequality) Thus, in $L^1_{loc}(\R)$ we can not define $\ee^{(\partial_x)^2}$ using series.}

First, we recall the Paley-Wiener theorem (see e.g. Theorem 19.3 in \cite{Rudin}), which says that for any
$u \in L^2(\R)$ with compact support the number
\begin{equation} \label{pw} b_u=\limsup_{y\rightarrow \infty}\log|{\mathcal F}(u)(\ii y)|/y \; .
\end{equation}
is finite and equals the right endpoint of the support of $u$. Similarly, the left endpoint $a_u$ of $\supp u$ is given by
$a_u=-\limsup_{y\rightarrow \infty}\log|{\mathcal F}(u)(-\ii y)|/y$.
Now, to prove the claim,
suppose the limit exists in $L^1_{loc}(\R)$ and denote it by $g$. Clearly $\supp f_K\subset \supp f=[a_f,b_f]$ and so
\begin{equation}\label{a80}
a_f\leq \supp g\leq b_f
\end{equation}
and $(f_K)_{K=1}^\infty$ converges to $g$ in $L^1(\R)$. Thus
\[
 \lim_{K \rightarrow \infty}\mathcal \f(f_K)(z)=
\lim_{K \rightarrow \infty} \int_a^b f_K(t) \ee^{-\ii z t}\frac{dt}{\sqrt{2\pi}}= \int_a^b g(t) \ee^{-\ii z t}\frac{dt}{\sqrt{2\pi}} =
{\mathcal F}(g)(z)\; .
\]
On the other hand, by partial integration we obtain
\begin{align*}
&{\mathcal F}(f_K)(z) = \int_a^b \sum_{k=0}^K \frac{\partial_t^{2k}f(t)}{k!}\,\ee^{-\ii z t} \frac{dt}{\sqrt{2\pi}}
=\\&= \sum_{k=0}^K \frac{(-z)^{2k}}{k!} \int_a^b f(t) \ee^{-\ii z t} \frac{dt}{\sqrt{2\pi}}=\sum_{k=0}^K \frac{(-z)^{2k}}{k!}\f(f)(z) \; ,
\end{align*}
and therefore
$$
{\mathcal F}(g)(z) = \ee^{-z^2}{\mathcal F}(f)(z).
$$
Combining (\ref{pw}) and (\ref{a80}) we have
\begin{align*}&
 b_f \geq b_g= \limsup_{y\rightarrow \infty} \frac{\log |{\mathcal F}(g)(\ii y)|}{y}
   = \limsup_{y \rightarrow \infty} \frac{\log | \ee^{y^2} {\mathcal F}(f) (\ii y) |}{y}=\\
   &=  \limsup_{y \rightarrow \infty} ~y+\frac{\log |{\mathcal F}(f) (\ii y) |}{y}=\infty+b_f=\infty .
\end{align*} This contradiction
proves the claim.

$\Box$

\end{ex}

In the above example we worked with the concrete function $\phi(z)=\ee^{z^2}$, which is an entire function of order 2. Based on deep results on 
entire functions, (see \cite{Levin}), it is possible to show that the same contradiction arises for any $\phi$ of order $>1$ with regular growth. 
On the other hand, it is well known that $\ee^{-\partial_x^2}f(x)$ appears in the solution of the heat equation on the line. In this case, 
$\ee^{- \partial_x^2}f(x)$ has the physically correct interpretation $\frac{1}{\sqrt{4\pi}} f*\ee^{-x^2/4}$, which at first seems to be a contradiction. 
This is not the case, as we further discuss in Sections \ref{sectd}, \ref{secCE} and \ref{secREW}. The connection with the heat equation is worked out 
in Example \ref{b1}. In the coming two sections we develop the general theory for solving (\ref{a1w}).

\section{Particular solutions}\label{secPart}
Using the machinery developed in the previous section, the existence of a solution to (\ref{a1w}) is immediate. Given an entire function $\phi$ we 
let $\z(\phi)$ denote the set of zeroes of $\phi$ and $|\z(\phi)|$
the set of their respective modulii.

\begin{theorem}\label{a2''}
Let $\phi$ be an entire function. The equation
\begin{equation}\label{a1'}
\phi\left(\frac{d}{dx}\right)f(x)=g(x), \quad g\in Exp,
\end{equation}
always has at least one solution. More precisely, if $g=\p(d_g)$ and $\supp d_g\cap\z(\phi)=\emptyset$, and then a solution
is given by
\begin{equation}\label{a14}
f=\p\left(\frac{d_g}{\phi}\right).
\end{equation}
\end{theorem}

\begin{proof}
Because of Theorem \ref{c1}, Equation (\ref{a1'}) is equivalent to
\[
\p (\phi d_f) = g \; ,
\]
in which the ``unknown'' is $d_f$. That (\ref{a14}) provides a solution is immediate (since multiplication of functions and distributions is associative, 
see \cite{H}).
\end{proof}

Note that by Proposition \ref{p1} we can always pick $d_g=\mu_g$ to be a measure supported on a circle of radius $R\not\in |\z(\phi)|$, so that the 
condition $\supp d_g\cap\z(\phi)=\emptyset$ is met. In particular, for
every $R\not \in |\z(\phi)|$ greater than an exponential bound $\tau$ of $g$ (recall Definition \ref{d1}), equation
(\ref{a14}) can be explicitly evaluated as
\begin{equation}\label{a14'}
f(x)=\int_{|\zeta|=R} \ee^{x\zeta}\frac{\B(g)(\zeta)}{\phi(\zeta)}~\frac{d\zeta}{2\pi \ii} \; ,
\end{equation}
that is, we take $\mu_g$ as in (\ref{mug}). By Cauchy's theorem and the remarks following Definition \ref{d5}, this argument can be taken one step 
further which we state as a separate corollary. 

\begin{corollary}\label{c7}
Let $\gamma$ be any simply connected closed curve avoiding $\z(\phi)$ with the conjugate diagram $S$ in its interior. Then \begin{equation*}
f(x)=\int_\gamma \ee^{x\zeta}\frac{\B(g)(\zeta)}{\phi(\zeta)}~\frac{d\zeta}{2\pi \ii}
\end{equation*}
solves (\ref{a1'}).
\end{corollary}

Sometimes it is however better to work with distributions than with the above integrals, as the the example below illustrates.

\begin{ex}\label{b5}
In the previous section we noted that
$$\sum p_k(x)\ee^{\zeta_k x}=\p\left(\sum p_k(-\partial_{\xi})\delta_{\zeta_k}\right)$$ for finite
sums. If the right hand side $g$ has this form and $d_g=\sum p_k(-\partial_{\xi})\delta_{\zeta_k}$, then (\ref{a14}) evaluates to
\begin{equation*}
f(x)=\p\left(\frac{\sum p_k(-\partial_{\xi})\delta_{\zeta_k}}{\phi}\right)=
\left\langle\ee^{x(\xi+\ii\eta)},\frac{\sum p_k(-\partial_{\xi})\delta_{\zeta_k}}{\phi}\right\rangle=
\sum\left\langle p_k(\partial_{\xi}) \frac{\ee^{x(\xi+\ii\eta)}}{\phi(\xi+\ii\eta)}, \delta_{\zeta_k}\right\rangle,
\end{equation*}
assuming that $\z(\phi)\cap \{\zeta_k\}=\emptyset$. If this is not the case, we are forced to work with the
more cumbersome expression (\ref{a14'}), or resort to some perturbation analysis. Let us assume that $\z(\phi)\cap \{\zeta_k\}=\emptyset$.
For concrete values of $p_k$ and $\phi$, the last expression above can clearly be evaluated explicitly. We see immediately that the solution $f$ 
will be of the form $f(x)=\sum q_k(x)\ee^{\zeta_k x}$ as well, for some $q_k\in Pol$.

$\Box$

\end{ex}

\section{Homogeneous solutions}\label{secHom}
An immediate question in connection with Theorem \ref{a2''} is whether different $d_g$'s give rise to the
same solution $f$. This is not the case, which is to be expected since if $\phi$ is a polynomial
(so that (\ref{a1'}) is simply a constant coefficient ordinary
differential equation) then the homogenous solutions are linear combinations of exponential functions $\ee^{\lambda x}$ where $\phi(\lambda)=0$.
This points out another pleasant fact when working with $Exp$, as opposed to traditional spaces like $L^2(\R)$ or Sobolev spaces, namely that 
homogenous solutions are inside of the space.

\begin{theorem}\label{c2}
Let $f\in Exp$ be a solution to $\phi\left(\frac{d}{dx}\right)f(x)=0$ of exponential type $\tau$.
Let $\{\zeta_k\}_{k=0}^\infty=\z(\phi)$ be an enumeration of $\z(\phi)$ and let $m_{k}$ denote the multiplicity of
$\zeta_k$. Then there are polynomials $p_{k}$ of degree $<m_{k}$ such that
\begin{equation}\label{a16}
f(x)=\sum_{ |\zeta_k|\leq\tau }p_k(x)\ee^{\zeta_k x} \; .
\end{equation}
\end{theorem}

Note that the sum in (\ref{a16}) is necessarily finite, since zeroes of analytic functions are discrete.
Also note that if the zeroes of $\phi$ have multiplicity 1 then (\ref{a16})
reduces to
$$
f(x)=\sum_{ |\zeta_k|\leq\tau }c_k\ee^{\zeta_k x},\quad c_k\in\C.
$$
The proof of the theorem is given at the end of this section. We first note the following immediate corollaries to Theorem \ref{a2''} and Theorem \ref{c2}.

\begin{corollary}\label{c6}
$\phi\left(\frac{d}{dx}\right)$ is invertible in $Exp$ if and only if the symbol $\phi$ has no zeroes. In this case,
$$\left(\phi\left(\frac{d}{dx}\right)\right)^{-1}=\frac{1}{\phi}\left(\frac{d}{dx}\right).$$
\end{corollary}

By Proposition \ref{p1} we also easily get
\begin{corollary}\label{c3}
Let $f_1=\p(\mu_1)$ and $f_2=\p(\mu_2)$ be solutions to the
non-homogeneous equation $(\ref{a1'})$, where the measures (or
distributions) $\mu_1,\mu_2$ have compact support. Set
$R=\sup\{|z|: ~z\in \supp\mu_1\cup\mu_2\}$. Then $f_1-f_2$ is
given by $(\ref{a16})$ with $\tau=R$.
\end{corollary}

To prove Theorem \ref{c2}, we first need a basic lemma. Fix $R>0$ and set $B_R=\{\zeta:~|\zeta|\leq R\}$. Let $A(B_R)$ denote the set of continuous 
functions that are analytic inside $B_R$, endowed with the supremum norm. Also let $E_z\in A(B_R)$ denote the function $E_z(\zeta) = \ee^{\zeta z}$, 
$\zeta \in B_R$. Given an entire function $\phi$ we set
$$\M_{\phi,R}=\cl\left(\Span\left\{ E_z\, \phi :~z\in\C \right\}\right),$$
where $\cl$ denotes the closure in $A(B_R)$.

\begin{lemma}\label{c4}
Let $\phi$ be an entire function and let $R\not \in |\z(\phi)|$. Let
$\{\zeta_k\}_{k=1}^K$ be an enumeration of $\z(\phi)\cap B_R$ and
let $m_k$ denote their corresponding multiplicities. Then
\begin{align*}
\M_{\phi,R}=\{\psi\in A(B_R):~\psi\text{ is zero at $\zeta_k$ with
multiplicity $\geq m_k$, }1\leq k\leq K\}.
\end{align*}
\end{lemma}
\begin{proof}
Since
point evaluations are continuous functionals on $A(B_R)$, it follows that $\M_{\phi,R}$ is included in the set to the right. Conversely, let $\psi$ be
in this set. Then $\psi/\phi\in A(B_R)$ and it is well known that $Pol$ is dense in $A(B_R)$. (To see this, first show that a given function 
$\tau\in A(B_R)$ can be approximated arbitrarily well by a dilation $\tau_r=\tau(r\cdot)$, $r<1$, and then use that the Taylor series of an analytic 
function converges uniformly on compacts). Thus, given
any $\epsilon>0$ there is a $p\in Pol$ with $\|\psi/\phi - p
\|_{A(B_R)}<\epsilon$ and it follows that $\psi\in
\cl\left(\phi Pol\right)$. Moreover,
$$
Pol\subset \cl\left(span\left\{ E_z:~z\in\C \right\}\right).
$$
To see this, note that the right hand side is an algebra which
contains 1 and $\zeta$, which is easily seen by considering the
limit of $(\ee^{z\zeta}-1)/z$ as $z\rightarrow 0$. Thus
$$
\psi\in \cl\left(\phi\, Pol \right)\subset
\cl\left( span \left\{ E_{z}\,\phi :~z\in\C \right\}\right) \; ,
$$
as desired.
\end{proof}

\textit{Proof of Theorem \ref{c2}.} Since $|\z(\phi)|$ is discrete we can find an $R>\tau$ such that there are no zeroes in $\tau<|\zeta|\leq R$. 
By Proposition \ref{p1}
we can write $f=\p(\mu)$ where $\mu$ is a measure supported on $\{\zeta:~|\zeta|=R\}$. By Theorem \ref{c1} we have $\p(\phi\mu)\equiv 0$. This
implies that $\mu$ defines a continuous linear functional on $A(B_R)$ which annihilates $\M_{\phi,R}$. By
Lemma \ref{c4} it follows that
$A(B_R)/\M_{\phi,R}$ is finite dimensional and that an equivalence class is uniquely specified by the values 
$$\left\{\frac{d^j}{d\zeta^j}\psi(\zeta_k):~1\leq k\leq K,~0\leq j<m_k\right\}.$$ 
It thus follows by linear algebra that for any linear functional $l$ on $A(B_R)/\M_{\phi,R}$ there are $p_k\in Pol$ of degree $<m_k$ such that
\begin{equation*}
l(\psi)=\sum_{k=1}^K p_k\left(\frac{d}{d\zeta}\right)\psi \Big|_{\zeta_k}.
\end{equation*}
In particular this is true for the functional induced by $\mu$ and thus
$$
f(x)=\p(\mu)(x)=\int \ee^{x\zeta}~d\mu(\zeta)=
\sum_{k=1}^K p_k\left(\frac{d}{d\zeta}\right)\ee^{x\zeta}\Big|_{\zeta_k}=
\sum_{k=1}^K p_k(x)\ee^{x\zeta_k},
$$
which is what we wanted to prove.

$\Box$

Note that if $\phi(z)=\sum_{k=0}^\infty a_k z^k$ has infinitely many zeroes $\{\zeta_k\}_{k=1}^\infty$,
then it is very likely that we can choose non-zero coefficients $c_k$ such that
$$
f(x)=\sum_{k=1}^\infty c_k\ee^{x\zeta_k}
$$
defines a solution to $\phi\left(\frac{d}{dx}\right)f(x)=0$ in some sense. However, the above function will not be in $Exp$, and hence this falls 
outside the analysis presented here. One way to work with larger spaces than $Exp$ in a general framework is to consider $\p(\mu)$ for measures $\mu$ 
that do not have compact support, but that decay fast enough that $$\p(\phi\mu)(z)=\int \ee^{z\zeta}\phi(\zeta)~d\mu(\zeta)$$ is convergent for all 
$z\in\C$. Such approach has been anticipated e.g. by Carmichael (see \cite{Car}, Sec. 4, references therein) and will be considered elsewhere. 

We will devote the remainder of this paper to examples (Sections \ref{sectd} and \ref{secCE}), a comparison with other methods (Section \ref{secREW}), 
a discussion of initial value problems (Section \ref{secini}) and the issue of how to practically implement the theory developed above 
(Section \ref{secImp}).

\section{Translation-differential equations}\label{sectd}
We consider some special classes of symbols $\phi$'s in this section and the next, beginning with functions of the form
\begin{equation}\label{er5}
\phi(z)=\sum_{k=1}^n p_k(z)\ee^{\zeta_k z},\quad p_k\in Pol,~\zeta_k\in\C.
\end{equation}
We call the corresponding equation $\phi(\frac{d}{dx})f=g$ a translation-differential equation, since
it is easily seen (recall (\ref{e4})) that the equation reduces to
\begin{equation}\label{a18}
\sum p_k\left(\frac{d}{dx}\right)f(x+\zeta_k)=g(x),
\end{equation}
and is thus a combination of classical ordinary differential equations and translation equations.

\begin{example}\label{b6}
Set $\phi(z)=2z\cosh(z)=z(\ee^{z}+\ee^{-z})$. Then (\ref{a18}) reduces to \begin{equation}\label{ex1}f'(x+1)+f'(x-1)=g(x).\end{equation}
Returning to the discussion at the end of Section \ref{secty}, This example shows that for \textit{particular choices of} $\phi$, $Exp$ may be 
unnecessarily restrictive as domain of the operator $\phi(\frac{d}{dx})$. For instance, $C^1(\R)$ would be a natural environment to work with
in the present case. 

\end{example}

$\Box$

We now consider an even simpler example. $L^2_{loc}(\R)$ will denote the set of functions which are in $L^2(I)$ when restricted to any compact 
interval $I\subset\R$.

\begin{example}\label{e2}
Put $\phi(z)=\ee^{z}-1$. The corresponding equation (\ref{a18}) with $g\equiv 0$ is $$f(x+1)-f(x)=0.$$ This equation makes perfect sense e.g. in 
$L^2_{loc}$, and clearly the solutions in this space are all functions with period 1.
Since $\z(\phi)=\{\ii 2\pi k\}_{k\in\Z}$, Theorem \ref{c3} shows that any solution in $Exp$ is of the
form $$\sum_{finite} c_k\ee^{\ii 2\pi k x}, \quad c_k\in\C,$$ which indeed is a function of period 1.
Now, if we allow infinite sums above, then, by standard Fourier-series on an interval, we can reach any
1-periodic function in $L^2_{loc}(\R)$. Returning to the discussion at the end of Section \ref{secHom}, we see that Theorem \ref{c2} can be used 
to find solutions to $\phi(\frac{d}{dx})f=0$ outside of $Exp$ given that $\phi(\frac{d}{dx})$ has a proper interpretation in the larger space in 
question. We will leave a formal study of when this is possible for a separate work.
\end{example}

$\Box$

\section{Convolution equations}\label{secCE}

As a second example of a particular class of symbols $\phi$'s, fix $u\in L^1(\R)$ with compact support and let us consider 
$\phi(z)=\int_{\R}\ee^{-z x} u(x)~dx$. Given $f\in Exp$ we then have 
\begin{equation}\label{tyh}
\begin{aligned}
&\phi(\frac{d}{dx})f(x)=\p(\phi \mu_f)(x)=\int_{\C} \ee^{x\zeta} \phi(\zeta)~d\mu_f(\zeta)=\\
&=\int_{\C} \ee^{x\zeta} \int_{\R}\ee^{-\zeta y}u(y)~dy~d\mu_f(\zeta)=\int_{\R}u(y)\int_{\C}\ee^{\zeta (x-y)}~d\mu_f(\zeta)~dy=\\
&=\int_{\R}u(y)f(x-y)~dy=u* f(x),
\end{aligned}\end{equation} 
where the use of Fubini's theorem is allowed due to the compact support of both $u$ and $\mu_f$. In other words, {\em all convolution equations where 
the convolver has compact support fits in the general theory developed in the earlier sections}. This example clearly demonstrates the non-local nature 
of $\phi(\frac{d}{dx})$, even when $\phi$ is a function of exponential type. For $\phi$ as above, one natural space in which to consider the equation 
\begin{equation}\label{aga}\phi(\frac{d}{dx})f=0\end{equation} 
would be $C(\R)$, the space of all continuous functions with the topology of uniform convergence on compacts. It is well known that its dual is the set 
of distributions with compact support. By Theorem \ref{c2} we easily get that the closure in $C(\R)$ of $\Span\{\ee^{\zeta x}\}_{\z(\phi)}$ satisfies 
$u*f=0$. That no other functions in $C(\R)$ satisfy this is true but difficult to show. It was proved by L. Schwartz, see \cite{schwartz}.

\begin{example}\label{e5}
We start with an easy example, namely $u=\chi_{[-1,0]}$, and suppose we wish to find all solutions to $u*f=0$ in $L^2_{loc}$, not just in $Exp$. Then 
$\phi(z)=\frac{\ee^{z}-1}{z}$ with zeroes $\{2\pi\ii k\}_{k\in\Z\setminus\{0\}}$, so the solutions in $Exp$ are 
$\{\ee^{2\pi\ii k x}\}_{k\in\Z\setminus\{0\}}$. Upon taking the closure in $L^2_{loc}$ we get all functions with period 1 that are orthogonal to 
$\chi_{[0,1]}$. That this set is precisely the set of all solutions to the original equation $\chi_{[-1,0]}*f=0$ is a consequence of Beurling's theorem 
(\cite{koosis}) and a short argument, we omit the details.
\end{example}

\begin{example}\label{e7}
Consider the equation \begin{equation}\label{e6}f=\Delta u*f+g\end{equation} in one variable, (i.e. $\Delta=\frac{d^2}{dx^2}$), where $f$ is the unknown 
function. If $u$ has compact support and $f\in Exp$ then
\begin{align*}
&\Delta u*f=\Delta (u*f)=\Delta_x \int_{\C} \ee^{x\zeta} \int_{\R}\ee^{-\zeta y}u(y)~dy~d\mu_f(\zeta)=\\
&= \int_{\C} \zeta^2\ee^{x\zeta} \int_{\R}\ee^{-\zeta y}u(y)~dy~d\mu_f(\zeta)=\p\left(\zeta^2\int_{\R}\ee^{-\zeta y}u(y)~dy~\mu_f(\zeta)\right).
\end{align*} 
Hence, setting $\phi(z)=z^2 \int_{\R}\ee^{-z x} u(x)~dx$, we have by Theorem \ref{c1} that
$$\Delta u*f=\p(\phi\mu_f)=\phi(\frac{d}{dx})f.$$ Setting $\psi=1-\phi$, we get that (\ref{e6}) is equivalent with $\psi(\frac{d}{dx})f=g$, which thus 
can be solved by Theorem \ref{a2''}, given that $g\in Exp$. Moreover, Theorem \ref{c2} tells us that the zeroes of $\psi$ determines the homogenous 
solutions in $Exp$ and, as before, we can retrieve more solutions by taking limits in some appropriate topology.
\end{example}

\section{Review of other methods and examples}\label{secREW}

We now compare the methods of this paper with other classical approaches to the development of a functional calculus for $\frac{d}{dx}$.
The standard approach is to notice that
$-\ii \frac{d}{dx}$ is an unbounded self-adjoint operator on e.g. $L^2(\R)$, and therefore the
spectral theorem provides us with a functional calculus. The corresponding spectral projection operator turns out to be the Fourier transform and 
hence, given any $\psi \in L^\infty(\R)$ we get (by definition)
\begin{equation}\label{a3}
\psi\left(-\ii\frac{d}{dx}\right)f(x)=\f^{-1}(\psi\f(f))=
\frac{1}{\sqrt{2\pi}}\int \psi(\xi)\widehat{f}({\xi})\ee^{\ii x\xi}~d\xi.
\end{equation}
If we allow $\psi$ to be unbounded, the same definition works but we get convergence
issues in the integral (\ref{a3}). There are various ways of dealing with this. The operator-theoretical
standpoint is that $\psi\left(-\ii\frac{d}{dx}\right)$ is a densely
defined operator, and hence the domain of definition is restricted to those $f$'s such that
$\psi(\xi)\widehat{f}({\xi})\in L^2(\R)$.  
The second possibility is to assume that $\psi$ is a Schwartz
distribution, $\psi\in\s'$, and $f\in \s$ (or vice versa). In this case, the integral in (\ref{a3}) is symbolic but
the theory is on solid ground. Moving one step further we can let $\psi$ depend on $x$ as well, and
we arrive at the theory of pseudo-differential operators, of which (\ref{a3}) is a
particular case. A third option is to work with functions whose Fourier transforms are functions or distributions with compact support, see for 
instance \cite{dybinskii,Tran}. This alternative is clearly the closest to the theory of the present paper, with the main difference being that 
the corresponding functions spaces are much smaller than $Exp$, since they contain no functions with exponential growth.

If $f\in Exp\cap L^2(\R)$ and if $\psi\in\s'$ is the restriction to $\R$ of an entire function, these two avenues coincide with the theory presented 
in this paper. More precisely, by the Paley-Wiener theorem we have that functions in $Exp\cap L^2(\R)$ are band-limited, i.e. $\f^{-1}(f)=\check{f}$ 
has compact support. By (\ref{fourier}) we then have
$$f(x)=\f(\check{f})(x)=\frac{1}{\sqrt{2\pi}}\p\Big(\delta_0(\xi)\check{f}(-\eta)\Big)(x)=\frac{1}{\sqrt{2\pi}}\p\Big(\delta_0(\xi)\hat{f}(\eta)\Big)(x).$$ With $\phi(z)=\psi(-\ii z)$ Theorem \ref{c1} yields
\begin{equation}\label{e3'}\begin{aligned}&\phi\left(\frac{d}{dx}\right)f(x)=\frac{1}{\sqrt{2\pi}}\p\Big(\phi(\zeta) \delta_0(\xi)\hat{f}(\eta)\Big)(x)=\frac{1}{\sqrt{2\pi}}\left\langle \ee^{\zeta x}\phi(\zeta),\delta_0(\xi)\hat{f}(\eta)\right\rangle=\\&=\frac{1}{\sqrt{2\pi}}\int_\R \ee^{\ii \eta x}\phi(\ii \eta) \hat{f}(\eta)~d\eta=\frac{1}{\sqrt{2\pi}}\int_\R \ee^{\ii \eta x}\psi(\eta) \widehat{f}(\eta)~d\eta,\end{aligned}\end{equation}
which is precisely (\ref{a3}). In particular, we note that the limit (\ref{a2'}) indeed exists. It is immediate that a solution to $\phi(\frac{d}{dx})f=g$ is given by the classical formula
\begin{equation}\label{a5}
{f}=\f^{-1}\left(\frac{\widehat{g}}{\psi}\right),
\end{equation}
given that the right hand side is well defined. This clearly runs into problems e.g. if $\psi$ has zeroes on $\R$, an obstacle which is completely absent 
in Theorem \ref{a2''} due to the flexibility in choosing $\mu_g$.

To further highlight some differences, recall Example \ref{b6}. We first remark that with the ``Fourier transform approach'' (\ref{a3}), the corresponding 
symbol is $\psi(x)=2\ii x \cos(x)$. However, for solving (\ref{a1'}) the formula (\ref{a5}) will not work due to the fact that $\cos(x)$ has zeroes
on $\R$. Despite this, there exists a multitude of solutions given by the theory in Sections
\ref{secPart} and \ref{secHom}. Of course, the issue of the zeroes may be overcome by using tricks, but we wish to demonstrate the simplicity and clarity of
the approach developed here. Moreover, we remark that if we depart from real exponents $\zeta_k$ in (\ref{a18}), the
Fourier transform approach \textit{is} out of the picture. 

We now take a look at a more intricate example which highlight both differences and similarities.

\begin{ex}\label{b1}
Consider the heat equation on $\R$, $$\left\{\begin{array}{l}
                                        (\partial_{t}-\partial_x^2)f(x,t)=0,\quad t>0 \\
                                        f(x,0)=g(x)
                                      \end{array}\right.
$$
where $g(x)$ is the initial heat distribution along $\R$. From a naive point of view, treating
$\partial_x^2$ as a number and recalling the theory of ordinary differential equations, the solution
should be
\begin{equation}\label{a6}
f(x,t)=\ee^{t \partial_x^2}g(x).
\end{equation}
Let $g(x)=\chi_{\R^+}$, i.e. the characteristic function of the positive real axis. Using the distribution
theory interpretation of (\ref{a6}), there is no problem defining $\widehat{g}$, and then (\ref{a3}) gives
\begin{equation}\label{a7}
f(x,t)=\frac{1}{\sqrt{2\pi}}\int \ee^{-t\xi^2}\widehat{g}({\xi})\ee^{\ii x\xi}~d\xi
\end{equation}
which is well defined since $\ee^{-t\xi^2}\in\s$. Recall that $\f^{-1}(\ee^{-t\xi^2})=\frac{1}{\sqrt{2t}}\ee^{-x^2/4t}$ and that the 
formula $\f^{-1}(\phi\psi)=\frac{1}{\sqrt{2\pi}}\check{\psi}*\check{\phi}$ also holds if $\psi\in\s$ and $\phi\in\s'$. Hence (\ref{a7}) 
can be recast as
$$f(x,t)=
\frac{1}{\sqrt{2\pi}}\f^{-1}(\ee^{-t\xi^2})*g=\frac{\ee^{-x^2/4t}}{2\sqrt{\pi t}}*\chi_{\R^+},$$
which is the well-known correct solution of the problem; $\frac{\ee^{x^2/4t}}{2\sqrt{\pi t}}$ is an
approximate identity as $t\rightarrow 0$ which satisfies the heat equation for every $t>0$, as is easily
verified. Moreover, the above solution corresponds with our physical intuition; if we have an isolated
infinite bar with temperature 1 on $\R^+$ and 0 on $\R^-$, then the temperature should eventually even
out to 1/2 over the entire bar, which is precisely what happens above.

On the other hand, if we try to interpret (\ref{a6}) via the limit (\ref{a2w}), we obviously get nonsense, for we
then have $\phi(z)=\sum_{k=0}^\infty \frac{t^k}{k!} z^{2k}$ but
$$\frac{\ee^{x^2/4t}}{2\sqrt{\pi}t}*\chi_{\R^+}\neq
\lim_{K\rightarrow\infty}\sum_{k=0}^K  \frac{t^k}{k!} \frac{d^{2k}}{dx^{2k}}\chi_{\R^+}$$
since the right-hand side for each fixed $K$ equals $\chi_{\R^+}$ plus a distribution with support at 0,
which in no sensible topology can converge to $1/2$.

At first sight, this seems to be in contrast with the equivalence of (\ref{a3}) and (\ref{e3'}). However, the complications arise from the 
fact that $\chi_{\R^+}$ does not lie in $Exp$. For $g\in Exp$, the calculations in Section \ref{secCE} shows that (\ref{a6}) reduces to 
\begin{equation}\label{tyg}f(x,t)=\frac{\ee^{-x^2/4t}}{2\sqrt{\pi}t}*g\end{equation} 
also with the functional calculus developed in this paper. More precisely, $\ee^{tz^2}=\int \ee^{-zx}\frac{\ee^{-x^2/4t}}{2\sqrt{\pi t}}$ 
so (\ref{a6}) follows by the calculation (\ref{tyh}).
\end{ex}
$\Box$

\section{On initial value problems}\label{secini}

The search for an appropriate substitute of initial value conditions, (i.e. conditions on a solution $f$
to $\phi\left(\frac{d}{dx}\right)f(x)=g(x)$ given at $x=0$ which uniquely determines $f$), is an important topic in various branches of theoretical physics, 
see for instance \cite{AV,B,BK1,EW,Moeller}.
We argue in this section that no such substitute exists that will work in general, although it is clearly possible for specific choices of $\phi$. 
For example, if $\phi$ has finitely many zeroes, say $N$ counted with multiplicity, then it clearly follows from Theorems \ref{a2''} and \ref{c2} that
\begin{equation}\label{gt57}\left\{\begin{array}{c}
                                     \phi(\frac{d}{dx})f=g,\quad g\in Exp \\
                                     \frac{d^{j-1}}{dx^{j-1}}f(0)=c_j, \quad j=1,\ldots,N
                                   \end{array}
\right.\end{equation}
has a unique solution $f$ in $Exp$ for all fixed $g$ and $c_1,\ldots, c_{N}\in\C$. A general framework including this case has been developed in 
\cite{GPR_JMP, GPR_CQG}. However, 
when $\phi$ has infinitely many zeroes we see that in order to specify $f$ we will need an infinite amount of conditions at $x=0$. It has been observed 
several times before \cite{EW,Moeller}, that if $f$ is analytic this is the same as
uniquely determining $f$ before solving the equation, (since the Taylor series of an entire function defines it uniquely). Another option, as is argued 
e.g. in \cite{B}, is to consider the "initial value problem" 
\begin{equation}\label{gt5'}\left\{\begin{array}{c}
                                     \phi(\frac{d}{dx})f=g \\
                                     f|_I=F
                                   \end{array}
\right.\end{equation} 
(where $f|_I$ means the restriction of $f$ to $I$). Again we run into trouble if we ask that the solution $f\in Exp$, as $F$ then uniquely determines 
$f$ (since we can get the Taylor series at a point in $I$ from $F$) and hence the equation $\phi(\frac{d}{dx})f=g$ again becomes redundant.

Thus, for the question of appropriate IVP's to make sense we have to consider non-entire functions. Let us thus suppose that $\phi$ is such that 
$\phi(\frac{d}{dx})f=g$ has a physically sensible interpretation in some space $\H$ which does not contain only analytic functions, as in Examples 
\ref{b6}, \ref{e2}, \ref{e5} and \ref{b1} and \cite{GPR_JMP, GPR_CQG}. For instance, in Example \ref{e2}, the particular equation considered there 
has a natural interpretation outside of $Exp$, for example in $L^2_{loc}$, and initial data on the interval $I=[0,1)$ uniquely specifies $f$, (by 
setting $f(x+n)=F(x)$ for all $x\in [0,1)$ and $n\in\Z$). Thus, supposing that $f$ is given on $[0,r]$ with $r<1$, the equation
is under-determined and for $r>1$ over-determined. However, the fact that in Example \ref{e2} there exists an interval of a precise length
such that the values of $f$ there uniquely determines $f$ on $\R$, is clearly related to the very simple
structure of the symbol $\phi(z)=\ee^z-1$, and it is bound to fail in general. For example, dividing $e^z-1$ by $z$, we obtain Example \ref{e5} 
(we can keep $\H=L^2_{loc}$). Then the appropriate length of $I$ is clearly still 1, but now the equation
\begin{equation*}\left\{\begin{array}{c}
                                     \phi(\frac{d}{dx})f=0 \\
                                     f|_I=1
                                   \end{array}
\right.\end{equation*}
has no solution, as was shown in Example \ref{e5}. As one takes into account $\phi$'s with less structure, we get more erratic behavior. 
We now take a look at an example where the symbol $\phi$ is such that \begin{equation}\label{eg6}\left\{\begin{array}{c}
                                     \phi(\frac{d}{dx})f=0 \\
                                     f|_I=F
                                   \end{array}
\right.\end{equation}
has a solution \textit{for all $I$ and $F\in L^2(I)$}, but where no interval $I$ is sufficient to \textit{uniquely} determine a solution.

\begin{example}\label{e3}
Let $2m>r_m>m$, $m\in\N$, be numbers such that $\frac{r_m}{r_{m'}}\not \in \mathbb{Q}$ for all pairs $m\neq m'$. Consider the set $$Z= \{r_mn:~{{m\in\N}},~n\in\Z\setminus\{0\}\}.$$
Given any $\rho>1$ we have
$$
\sum_{\zeta\in Z}\frac{1}{|\zeta|^\rho}=
\sum_{n\in\Z\setminus\{0\}} \frac{1}{n^\rho} \sum_{m\in\N}\frac{1}{r_m^\rho}\leq  2\sum_{n\in\N} \frac{1}{n^\rho} \sum_{m\in\N}\frac{1}{m^\rho} <\infty.
$$
By the theory of Hadamard decompositions of analytic functions, there exists an entire function $\phi$
of order 1 with zeroes $Z$ (each with multiplicity 1). By Theorem \ref{c2} the solutions to the equation 
\begin{equation}\label{gt5}\phi(\frac{d}{dx})f=0\end{equation} are given by
\begin{equation}\label{a19}
\Span \cup_{\zeta\in Z} \{\ee^{\zeta x}\},
\end{equation}
where $\Span$ means taking all possible finite linear combinations. However, we will again assume that the nature of the problem (\ref{gt5}) is such 
that we can take limits in $L^2_{loc}$ and still get relevant solutions. We then claim that given any interval $I$ and any $F\in L^2(I)$, the equation 
system (\ref{eg6}) always have infinitely many solutions.

To prove this, we need to recall the classical theorem of Levinson on completeness of bases of exponential functions. Given a discrete set 
$\tilde{Z}\subset \C$ that does not contain 0, we define $n_{\tilde{Z}}(r)=\#\{\zeta\in\tilde{Z}:|\zeta|<r\}$, ($\#$ denotes the number of points 
in a set), and set $$N_{\tilde{Z}}(r)=\int_0^r \frac{n_{\tilde{Z}}(t)}{t}~dt. $$ Given an interval $I$, Levinson's theorem says that 
$$\Span\{\ee^{\zeta x}\}_{\zeta\in\tilde{Z}}$$ is dense in $L^2(I)$ if 
\begin{equation}\label{ju}\limsup_{r\rightarrow\infty} N_{\tilde{Z}}(r)-\frac{|I|r}{\pi}+\frac{\ln r}{2}>-\infty,
\end{equation} 
where $|I|$ denotes the length of $I$ (see e.g. \cite{Boas} or \cite{Young}). We will show that (\ref{ju}) is satisfied for $\tilde{Z}=Z$ and any $I$. 
Assuming this for the moment, our example is complete; for given $F$ on $I$ we can pick a larger interval $J$ and extend $F$ in different ways to $J$. 
Since all these extensions then lie in the $L^2(J)$-closure of $\Span\{\ee^{\zeta x}\}_{\zeta\in Z}$, they are all solutions to (\ref{gt5}) in the 
weaker sense considered here.

Now, (\ref{ju}) follows if we show that given any $C>0$ there exists an $R$ with 
\begin{equation}\label{rf}N_{Z}(r)\geq Cr,\quad r>R.\end{equation}
Set $Z_{m}=\{r_mn:~n\in\Z\setminus\{0\}\}$ so that $Z=\cup_{m\in\N}Z_m$. Then $$n_{Z_m}(r)\geq \frac{2r}{r_m}-2$$ so 
$$N_{Z_m}(r)=\int_{r_m}^r \frac{n_{Z_m}(x)}{x}~dx\geq \int_{r_m}^r \frac{2}{r_m }-\frac{2}{x}~dx=2\left(\frac{r}{r_m}-1-\ln\left(\frac{r}{r_m}\right)\right)$$
for $r>r_m$. It is easy to show that $\ln (x)\leq x/\ee$ so we can continue to get
$$N_{Z_m}(r)\geq\left(2-\frac{2}{\ee}\right)\frac{r}{r_m}-2\geq \left(2-\frac{2}{\ee}\right)\frac{r}{2m}-2.$$ Hence, given any $M\in N$ and $r$ large enough, 
we have $$N_{Z}(r)\geq \sum_{m=1}^M N_{Z_m}(r)\geq \left(2-\frac{2}{\ee}\right)\left(\sum_{m=1}^M\frac{1}{m}\right) r -2M$$ and (\ref{rf}) immediately 
follows since $\sum_{m=1}^M\frac{1}{m}=\infty$. 
\end{example}

$\Box$

The conclusion of Example \ref{e3} seems to imply that \textit{all} functions are solutions to $\phi(\frac{d}{dx})f=0$ as long as we restrict attention 
to a finite interval of arbitrary length. This actualizes the question of whether it makes any sense to study the equation $\phi(\frac{d}{dx})f=g$ for 
general entire functions $\phi$. We end this section with some comments on this. The discussion requires some acquaintance with the theory of entire 
functions, for which we refer to \cite{Boas}. It is easily seen that the symbols $\phi$ in Sections \ref{sectd} and \ref{secCE} are all of order 1 and 
\textit{finite type}. The $\phi$ in Example \ref{e3} still has order 1 but will have \textit{infinite type.} Clearly, the situation in Example \ref{e3} 
is related to $\phi$ having "too many zeroes". The below proposition shows that the phenomenon in Example \ref{e3} occurs for all symbols $\phi$ whose 
exponent of convergence is greater than 1. We note that this includes the majority of all entire functions of order $>1$, since the order and the exponent 
of convergence coincides if the \textit{canonical factor} is dominant in the Hadamard decomposition, (see Theorem 2.6.5 and Sec. 2.7 of \cite{Boas}). 
A notable exception to this is the symbol $\phi(z)=\ee^{z^2}$ considered earlier, which has order two and no zeroes. We remind the reader that the exponent 
of convergence is $>1$ if there exists a $\kappa>1$ such that 
\begin{equation}\label{p0}\sum_{\zeta\in\z(\phi)}\frac{1}{|\zeta|^\kappa}=\infty.\end{equation}

\begin{proposition}\label{tgb}
Let $\phi$ be an entire function such that the exponent of convergence of the zeroes of $\phi$ is $>1$. Then the solutions to $\phi(\frac{d}{dx}) f=0$ 
in $Exp(\R)$ are dense in $L^2(I)$ for any bounded interval $I\subset\R$.
\end{proposition}
\begin{proof}
There is no restriction to assume that $0\not \in \z(\phi)$. It suffices to show that 
\begin{equation}\label{t5}\limsup \frac{N_{\z(\phi)}(r)}{r}=\infty,\end{equation} 
for then (\ref{ju}) holds and hence Levinson's theorem does the job just as in Example \ref{e3}. Suppose that (\ref{t5}) is false. 
Then $N_{\z(\phi)}(r)\leq Cr$ for some $C>0$. Moreover 
\begin{equation}\label{eq0229}\begin{aligned}&n_{\z(\phi)}(r)=n_{\z(\phi)}(r)\ln\ee=n_{\z(\phi)}(r)\int_{r}^{\ee r}\frac{1}{x}~dx\leq\\&
\int_{r}^{\ee r}\frac{n_{\z(\phi)}(x)}{x}~dx\leq N_{\z(\phi)}(\ee r)\leq C\ee r\end{aligned}
.\end{equation}
However, with $\kappa>1$ in (\ref{p0}), Lemma 2.5.5 of \cite{Boas} implies that $$\infty=\int_0^\infty\frac{n_{{\z(\phi)}}(r)}{r^{\kappa+1}} ~dr $$ 
which contradicts (\ref{eq0229}), since $\int_1^\infty \frac{C\ee r}{r^{k+1}}~dr<\infty$.
\end{proof}

\section{Implementation}\label{secImp}

The formula (\ref{a14'}) for solving (\ref{a1'}) is complete from a theoretical perspective, and if the particular application
is such that $g$ is given by an analytical expression, it is perfectly possible to implement it in a
computer. However, if $g$ stems from measured data, (\ref{a14'}) is completely useless. For calculating
derivatives is a highly unstable operation on a measured signal. Calculating $g^{(10)}$ with some acceptable
 accuracy is usually out of the picture, whereas in order to get a good approximation of $\B(g)$ via
(\ref{a10}) we clearly need more than 10 Taylor coefficients. The obvious alternative is to try to get
$\B(g)$ using (\ref{a11}), i.e.
\begin{equation}\label{a11'}
\B(g)(\zeta)=\overline{\vartheta}\int_{0}^\infty g(t\vartheta)\ee^{-\zeta\overline{\vartheta} t}~dt
\end{equation}
but here the obstacle is that (\ref{a11}) can only be evaluated for $\vartheta=1$ or $\vartheta=-1$, since our measurements likely take place in 
the real world! With
\begin{equation}\label{a22}
X_+=\inf\{\xi:~g(t)\ee^{-\xi t}\in L^1(\R^+)\}
\end{equation}
and
\begin{equation}\label{a22'}
X_-=\sup\{\xi:~g(-t)\ee^{\xi t}\in L^1(\R^+)\}
\end{equation}
this will give $\B(g)(\zeta)$ on the domain $$\{\xi+\ii \eta:~\xi >X_+\text{ or } \xi <X_-\},$$ and hence
we find that (\ref{a14'}) can not be used since we lack information on the strip $$\{\xi+\ii \eta:~ X_-\leq \xi \leq X_+\}.$$

The following theorem basically says that if $\phi$ is bounded below on this strip, then (\ref{a1'}) can be solved using the Fourier transform. 
For clarity we indicate the independent variable upon which the Fourier transform acts with a subscript.

\begin{theorem}\label{c5}
Suppose that $\phi$ satisfies
$$\inf\left\{|\phi(\xi+\ii\eta)|:~\xi\in[\xi_-,\xi_+],~|\eta|>Y_0\right\}=\epsilon>0$$
for some $\xi+>X_+$, $\xi_-<X_-$ and $\epsilon,Y_0>0$. Then, given $g\in Exp$, the expression
\begin{equation}\label{a20}\begin{aligned}
f(x)=&\ee^{x\xi_+}\f^{-1}_s\left(\frac{\f_t\left(\chi_{\R^+}\ee^{-\xi_+ t}g\right)(s)}{\phi(\xi_++\ii s)}\right)(x)+\\&+\ee^{x\xi_-}\f^{-1}_s\left(\frac{\f_t\left(\chi_{\R^-}\ee^{-\xi_- t}g\right)(s)}{\phi(\xi_-+\ii s)}\right)(x)\end{aligned}
\end{equation}
provides a solution to $\phi\left(\frac{d}{dx}\right)f(x)=g(x)$ in $Exp$.
\end{theorem}

We remark that, given a sampling of the function $g$, the corresponding discrete version of (\ref{a20})
can be efficiently implemented using the Fast Fourier Transform \cite{FFT}.

\begin{proof}
By Corollary \ref{c7} we have that
\begin{equation}\label{a21}
f(x)=\int_{\Box_Y} \ee^{x\zeta}\frac{\B(g)(\zeta)}{\phi(\zeta)}~\frac{d\zeta}{2\pi \ii}.
\end{equation}
gives a solution, where $\Box_Y$ represents the rectangle with corners $\{\xi_+\pm\ii Y, \xi_-\pm\ii Y \}$ and $Y>Y_0$ is a parameter sufficiently large
that the sides are outside the conjugate diagram $S$ of $\B(f)$, and such that such that $\Box_Y\cap \z(\phi)$.
By the assumption on $\phi$ and Cauchy's theorem, (\ref{a21}) does not depend on $Y$, and hence we can take a limit as $Y\rightarrow\infty$. By
the equivalence of (\ref{a11'}), it is easy to see that the two vertical integrals
evaluate to
$$
\ee^{x\xi_+}\int_{-Y}^{Y} \ee^{\ii\eta x}
\frac{\f_t\left(\chi_{\R^+}\ee^{-\xi_+ t}g\right)(\eta)}{\phi(\xi_++\ii\eta)}~\frac{d\eta}{\sqrt{2\pi}}+\ee^{x\xi_-}\int_{-Y}^{Y} \ee^{\ii\eta x}\frac{\f_t\left(\chi_{\R^-}\ee^{-\xi_- t}g\right)(\eta)}{\phi(\xi_-+\ii\eta)}~\frac{d\eta}{\sqrt{2\pi}},$$
whose limit as $Y\rightarrow\infty$ equals (\ref{a20}). Hence (\ref{a20}) follows if we prove that the
horizontal integrals converge to 0 as $Y\rightarrow\infty$. We consider the upper one.
For each fixed $x\in\R$ we have
\begin{align*}&\left|-\int_{\xi_-}^{\xi_+} \ee^{x(\xi+\ii Y)}\frac{\B(g)(\xi+\ii Y)}{\phi(\xi+\ii Y)}~\frac{d\xi}{2\pi \ii}\right|\leq\\& \leq\frac{1}{2\pi\epsilon} \sup_{\xi_-<\xi<\xi_+} \left|\ee^{x\xi}\B(g)(\xi+\ii Y)\right|\leq const\cdot \sup_{\xi_-<\xi<\xi_+} |\B(g)(\xi+\ii Y)|.\end{align*}
Since $\B(g)$ is defined by a Laurant-series starting at $\zeta^{-1}$, see (\ref{a10}), it easily follows
that $$|\B(g)(\zeta)|\leq const/|\zeta|$$ for large $|\zeta|$, and hence
$\lim_{Y\rightarrow\infty} \sup_{\xi_-<\xi<\xi_+} |\B(g)(\xi+\ii Y)|=0,$ as desired.
\end{proof}

An interesting note is that although the above theorem was derived under the assumption that $g\in Exp$,
the formula (\ref{a20}) works for any function $g$ such that $X_+$ and $X_-$, as defined in (\ref{a22})
and (\ref{a22'}), are finite. For $g$ in $Exp$, $f$ given by (\ref{a20}) is a solution in the sense of
(\ref{a2w}). However, as was observed in Example \ref{b1}, this interpretation of our equation (\ref{a2w})
may be too restrictive for the underlying physical reality, and hence it is certainly worthwhile to
consider formula (\ref{a20}) even for examples where (\ref{a2w}) is too restrictive.

Theorem \ref{c5} typically applies to e.g. the translation-differential equations considered in Section \ref{sectd}, but it never applies to 
the convolution equations considered in Section \ref{secCE}. However, if $g$ e.g. is a band limited $L^2$ function (i.e. a function in the 
Paley-Wiener space), then we have $X_+=X_-=0$ and hence $\B(g)$ can be evaluated at all points of $\Box_Y$ if $Y$ is large enough. 
Thus (\ref{a21}) provides a computable solution as is, given that $Y$ is taken large enough.

\paragraph{\bf Acknowledgments}
M. Carlsson was supported by DICYT-USACH and STINT; H. Prado's research was partially supported by DICYT-USACH\# 041133PC and FONDECYT\# 1130554; 
E.G. Reyes' research was partially supported by FONDECYT grant \# 1111042.

Matematikcentrum, Lund University\\
Box 118, 221 00 Lund, Sweden\\ E-mail: Marcus.Carlsson@math.lu.se\\

\smallskip

Departamento de Matem\'atica y Ciencia de la
Computaci\'on, \\
Universidad de Santiago de Chile\\
Casilla 307 Correo 2, Santiago, Chile\\
E-mail: humberto.prado@usach.cl ; ereyes@fermat.usach.cl ; e\_g\_reyes@yahoo.ca
\end{document}